\documentclass[runningheads]{llncs}

\usepackage[utf8]{inputenc}
\usepackage{graphicx}
\usepackage{complexity}
\usepackage{xifthen}
\usepackage{amsmath}
\usepackage{amssymb}
\usepackage{amsfonts}

\newcommand{\var}{\mathsf{var}}
\newcommand{\size}{\mathsf{size}}

\newcommand{\MC}{\mathsf{MC}}
\newcommand{\HW}{\mathsf{HW}}
\newcommand{\CE}{\mathsf{CE}}

\newcommand{\decDNNF}{\mathsf{dec}\textsf{-}\mathsf{DNNF}}

\newcommand{\kcps}{\mathsf{kcps}}

\newcommand{\UNSAT}{\mathsf{UNSAT}}

\renewcommand{\SAT}[1][]{\ifthenelse{\isempty{#1}}{}{#1\textsf{-}}\mathsf{SAT}}
\newcommand{\sSAT}[1][]{\#\SAT[#1]}
\newcommand{\mSAT}[1][]{\mathsf{Max}\SAT[#1]}

\newtheorem{observation}{Observation}

\begin{document}
\title{Knowledge compilation languages as proof systems}
%
%
\author{Florent Capelli\inst{1}\orcidID{0000-0002-2842-8223}}
\authorrunning{F. Capelli}
\institute{Université de Lille, Inria, UMR 9189 - CRIStAL - Centre de Recherche en Informatique Signal et Automatique de Lille, F-59000 Lille, France  \\
\email{florent.capelli@univ-lille.fr}}
\maketitle              
\begin{abstract}
  In this paper, we study proof systems in the sense of Cook-Reckhow for
  problems that are higher in the polynomial hierarchy than coNP, in particular,
  \#SAT and maxSAT. We start by explaining how the notion of Cook-Reckhow proof
  systems can be apply to these problems and show how one can twist existing
  languages in knowledge compilation such as decision DNNF so that they can be
  seen as proof systems for problems such as \#SAT and maxSAT.

\end{abstract}

\section{Introduction}
\label{sec:intro}

Proof complexity studies the hardness of finding a certificate that a CNF
formula is not satisfiable. A minimal requirement for such a certificate is that
it should be checkable in polynomial time in its size, so that it is easier for
an independent checker to assess the correctness of the proof than to redo the
computation made by a solver. While proof systems have been implicitly used for
a long time starting with resolution~\cite{DavisP60,DavisPLL62}, their
systematic study has been initiated by Cook and Reckhow~\cite{cook1979relative}
who showed that unless $\NP=\coNP$, one cannot design a proof system where all
unsatisfiable CNF have short certificates. Nevertheless, many unsatisfiable CNF
may have short certificates if the proof system is powerful enough, motivating
the study of how such systems, such as resolution~\cite{DavisPLL62} or
polynomial calculus~\cite{Clegg96}, compares in terms of succinctness
(see~\cite{nordstrom2013pebble} for a survey). More recently, proof sytems found
practical applications as SAT solvers are expected -- since 2013 -- to output
proof of unsatisfiability in SAT competitions to avoid implementation bugs.

While the proof systems implicitly defined by the execution trace of modern CDCL
SAT solvers is fairly well understood~\cite{pipatsrisawat2011power}, it is not
the case for tools solving harder problems on CNF formulas such as $\sSAT$ and
$\mSAT$. For $\mSAT$, a resolution-like system for $\mSAT$ has been proposed by
Bonet et al.~\cite{BonetLM07} for which a compressed version has been used in a
solver by Bacchus and Narodytska~\cite{narodytska2014} but it is to the best of
our knowledge the only such proof system. To the best of our knowledge, no proof
system has been proposed for $\sSAT$.

In this short paper, we introduce new proof systems for $\sSAT$ and $\mSAT$.
Contrary to the majority of proof systems for $\SAT$, our proof systems are not
based on the iterative application of inference rules on the original CNF
formula. In our proof systems, our certificates are restricted Boolean circuits
representing the Boolean function computed by the input CNF formula. These
restricted circuits originate from the field of knowledge
compilation~\cite{DarwicheM2002}, whose primary focus is to study the
succinctness and tractability of representations such as Read Once Branching
Programs~\cite{Wegener00} or deterministic DNNF~\cite{Darwiche01MC} and how CNF
formula can be transformed into such representations. To use them as
certificates for $\sSAT$, we first have to add some extra information in the
circuit so that one can check in polynomial time that they are equivalent to the
original CNF. The syntactic properties of the input circuits then allow to
efficiently count the number of satisfying assignments, resulting in the desired
proof system. Moreover, we observe that most tools doing exact model counting
are already implicitly generating such proofs. Our result generalizes known
connections between regular resolution and Read Once Branching Programs
(see~\cite[Section 18.2]{Jukna12}).

The paper is organized as follows. Section~\ref{sec:preliminaries} introduces
all the notions that will be used in the paper. Section~\ref{sec:kcpc} contains
the definition of certified $\decDNNF$ that allows us to define our proof
systems for $\sSAT$ and $\mSAT$.

  

\section{Preliminaries}
\label{sec:preliminaries}

\paragraph{Assignments and Boolean functions.} Let $X$ be a finite set of
variables and $D$ a finite domain. We denote the set of functions from $X$ to
$D$ as $D^X$. An \emph{assignment on variables $X$} is an element of
$\{0,1\}^X$. A \emph{Boolean function $f$ on variables $X$} is an element of
$\{0,1\}^{\{0,1\}^X}$, that is, a function that maps an assignment to a value in
$\{0,1\}$. An assignment $\tau \in \{0,1\}^X$ such that $f(\tau) = 1$ is called
a \emph{satisfying assignment} of $f$, denoted by $\tau \models f$. We denote by
$\bot_X$ the Boolean function on variables $X$ whose value is always $0$. Given
two Boolean functions $f$ and $g$ on variables $X$, we write $f \Rightarrow
g$ if for every $\tau$, $f(\tau) \leq g(\tau)$.

\paragraph{CNF.} Let $X$ be a set of variable. A \emph{literal} on variable $X$
is either a variable $x \in X$ or the negation $\neg x$ of a variable $x \in X$.
A \emph{clause} is a disjunction of literals. A \emph{conjunctive normal form
  formula}, CNF for short, is a conjunction of clauses. A CNF naturally defines a Boolean
function on variables $X$: a \emph{satisfying assignment} for a CNF $F$ on
variable $X$ is an assignment $\tau \in \{0,1\}^X$ such that for every
clause $C$ of $F$, there exists a literal $\ell$ of $C$ such that $\tau(\ell) =
1$ (where we define $\tau(\neg x) := 1-\tau(x)$). We often identify a CNF with
the Boolean function it defines.

The problem $\SAT$ is the problem of deciding, given a CNF formula $F$, whether
$F$ has a satisfying assignment. It is the generic $\NP$-complete
problem~\cite{cook1971complexity}. The problem $\UNSAT$ is the problem of
deciding, given a CNF formula $F$, whether $F$ does not have a satisfying
assignment. It is the generic $\coNP$-complete problem.

Given a CNF $F$, we denote by $\#F = |\{\tau \mid \tau \models F\}|$ the number
of solutions of $F$ and by $M(F) = \max_{\tau} |\{C \in F \mid \tau \models
C\}|$ the maximum number of clauses of $F$ that can be simultaneously satisfied.
The problem $\sSAT$ is the problem of computing $\#F$ given a CNF $F$ as input
and the problem $\mSAT$ is the problem of computing $M(F)$ given a CNF $F$ as
input.

\paragraph{Cook-Reckhow proof systems.} Let $\Sigma, \Sigma'$ be finite
alphabets. A \emph{(Cook-Reckhow) proof system}~\cite{cook1979relative} for a
language $L \subseteq \Sigma^*$ is a surjective polynomial time computable
function $f : \Sigma' \rightarrow L$. Given $a \in L$, there exists, by
definition, $b \in \Sigma'$ such that $f(b) = a$. We will refer to $b$ as being
\emph{a certificate of $a$}.


In this paper, we will mainly be interested in proof systems for the problems
$\sSAT$ and $\mSAT$, that is, we would like to design polynomial time verifiable
proofs that a CNF formula has $k$ solutions or that at most $k$ clauses in the
formula can be simultaneously satisfied. For the definition of Cook-Reckhow,
this could translate to finding a proof system for the languages $\{(F, \#F)
\mid F \text{ is a CNF}\}$ and $\{(F, M(F)) \mid F \text{ is a CNF}\}$.

For example, a naive proof system for $\sSAT$ could be the following: a
certificate that $F$ has $k$ solutions would be the list of the $k$ solutions
together with a resolution proof that $F'=F \wedge \bigwedge_{\tau \mid \tau
  \models F} C_\tau$ is not satisfiable where $C_\tau := \bigvee_{\{x \mid
  \tau(x) = 0\}} x \vee \bigvee_{\{\neg x \mid \tau(x) = 1\}} \neg x$ is the
clause such that the only non-satisfying assignment is $\tau$. One could then
check in polynomial time that each of the $k$ assignments satisfies $F$ and that
$F'$ is indeed unsatisfiable and then output $(F,k)$. This proof system is
however not very interesting as one can construct very simple CNF with
exponentially many solutions: for example the empty CNF on $n$ variables has
$2^n$ and will thus have a certificate of size at least $2^n$.

\paragraph{$\decDNNF$.} A \emph{decision Decomposable Negation Normal Form
  circuit} $D$ on variables $X$, $\decDNNF$ for short, is a directed acyclic
graph (DAG) having exactly one node of indegree $0$ called the \emph{source}.
Nodes of outdegree $0$ are called the \emph{sinks} and are labeled by $0$ or
$1$. The other nodes have outdegree $2$ and can be of two types:
\begin{itemize}
\item The \emph{decision nodes} are labeled with a variable $x \in X$. One
  outgoing edge is labeled with $1$ and the other by $0$, represented
  respectively as a solid and a dashed edge in our figures.
\item The \emph{$\wedge$-nodes} are labeled with $\wedge$.
\end{itemize}
Moreover, we have two other syntactic properties. We introduce a few notations
before explaining them. If there is a decision node in $D$ labeled with variable
$x$, we say that $x$ is \emph{tested} in $D$. We denote by $\var(D)$ the set of
variables tested in $D$. Given a node $\alpha$ of $D$, we denote by $D(\alpha)$
the $\decDNNF$ whose source is $\alpha$ and nodes are the nodes that can be
reached in $D$ starting from $\alpha$. We also assume the following:
\begin{itemize}
\item Every $x \in X$ is tested at most once on every source-sink path of $D$.
\item Every $\wedge$-gate of $D$ are \emph{decomposable}, that is, for every
  $\wedge$-node $\alpha$ with successors $\beta,\gamma$ in $D$, it holds that
  $\var(D(\beta)) \cap \var(D(\gamma)) = \emptyset$.
\end{itemize}
\begin{figure}
  \includegraphics{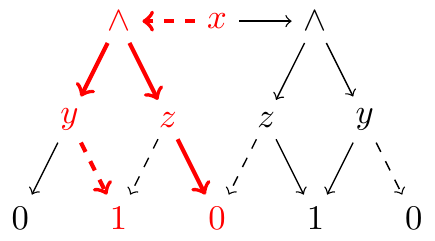}
  \centering
  \caption{A $\decDNNF$ computing $x = y = z$.}
  \label{fig:afbdd}
\end{figure}

Let $\tau \in \{0,1\}^X$. A source-sink path $P$ in $D$ is \emph{compatible}
with $\tau$ if and only if when $x$ is tested on $P$, the outgoing edge labeled
with $\tau(x)$ is in $P$. We say that $\tau$ satisfies $D$ if only $1$-sinks are
reached by paths compatible with $\tau$. A $\decDNNF$ and the paths compatible
with the assignment $\tau(x) = \tau(y) = 0, \tau(z) = 1$ are depicted in bold
red on Figure~\ref{fig:afbdd}. Observe that a $0$-sink is reached so $\tau$ does
not satisfy $D$. We will often identify a $\decDNNF$ with the Boolean function
it computes.

\begin{observation}
  \label{obs:patharecompatible} Given a $\decDNNF$ $D$ on variables $X$ and a
  source-sink path $P$ in $D$, there exists $\tau \in \{0,1\}^X$ such that $P$
  is compatible with $\tau$. Indeed, by definition, every variable $x \in X$ is
  tested at most once in $P$, thus, if $x$ is tested on $P$ in a decision node
  $\alpha$ and $P$ contains the outgoing edge labeled with $v_x$, we can choose
  $\tau(x) := v_x$. The value of $\tau$ for a variable $x$ not tested on $P$ can
  be chosen arbitrarily.
\end{observation}

The size of a $\decDNNF$ $D$, denoted by $\size(D)$ is the number of edges of the
underlying graph of $D$.


\paragraph{Tractable queries.} The main advantage of representing a Boolean
function with a $\decDNNF$ is that it makes the analysis of the function easier.
Given a $\decDNNF$, one can easily find a satisfying assignment by only following
paths backward from $1$-sinks. Similarly, one can also count the number of
satisfying assignments or find one satisfying assignment with the least number
of variables set to $1$ etc. The relation between the queries that can be solved
efficiently and the representation of the Boolean function has been one focus of
Knowledge Compilation. See~\cite{DarwicheM2002} for an exhaustive study of
tractable queries depending on the representation. Let $f : 2^X \rightarrow
\{0,1\}$ be a Boolean function. In this paper, we will mainly be interested in
solving the following problems:
\begin{itemize}
\item Model Counting Problem ($\MC$): return the number of satisfying assignment
  of $f$.
\item Clause entailment ($\CE$): given a clause $C$ on variables $X$, does $f
  \Rightarrow C$?
\item Maximal Hamming Weight ($\HW$): given $Y \subseteq X$, compute \[\max_{\tau
    \models f} |\{y \in Y \mid \tau(y)=1\}|.\]
\end{itemize}
All these problems are tractable when the Boolean function is given as a $\decDNNF$:

\begin{theorem}[\cite{Darwiche01MC,KoricheBLM16}]
  \label{thm:count_afbdd}
  Given a $\decDNNF$ $D$, one can solve problems $\MC, \CE, \HW$ on the Boolean
  function represented by $D$ in linear time in $\size(D)$.
\end{theorem}

The tractability of $\CE$ on $\decDNNF$ has the following useful consequence:
\begin{corollary}
  \label{cor:cnfentailement} Given a $\decDNNF$ $D$ and a CNF formula $F$, one
  can check in time $O(\size(F) \times \size(D))$ whether $D \Rightarrow F$.
\end{corollary}
\begin{proof}
  One simply has to check that for every clause $C$ of $F$, $D \Rightarrow C$,
  which can be done in polynomial time by Theorem~\ref{thm:count_afbdd}. 
\end{proof}

\section{Knowledge compilation based proof systems}
\label{sec:kcpc}

Theorem~\ref{thm:count_afbdd} suggests that given a CNF $F$, one could use a
$\decDNNF$ $D$ computing $F$ as a certificate for $\sSAT$. The proof system
could then check the certificate as follows:
\begin{enumerate}
\item\label{enu:count} Compute the number $k$ of satisfying assignments of $D$.
\item\label{enu:equiv} Check whether $F$ is equivalent to $D$.
\item If so, return $(F,k)$. 
\end{enumerate}
While Step~\ref{enu:count} can be done in polynomial time by
Theorem~\ref{thm:count_afbdd}, it turns out that Step~\ref{enu:equiv} is not
tractable:
\begin{theorem}
  \label{thm:equiv_intrac} The problem of checking, given a CNF $F$ and an
  $\decDNNF$ $D$ as input, whether $F \Rightarrow D$ is $\coNP$-complete.
\end{theorem}
\begin{proof}
  The problem is clearly in $\coNP$. 
%
  For completeness, there is a straightforward reduction to $\UNSAT$. Indeed,
  observe that a CNF $F$ on variables $X$ is not satisfiable if and only if $F
  \Rightarrow \bot_X$. Moreover, $\bot_X$ is easily represented as a $\decDNNF$
  having only one node: a $0$-labeled sink.
\end{proof}




\subsection{Certified $\decDNNF$}
\label{sec:certdnnf}

The reduction used in the proof of Theorem~\ref{thm:equiv_intrac} suggests that
the $\coNP$-completeness of checking whether $F \Rightarrow D$ comes from the
fact that $\decDNNF$ can succinctly represent $\bot$. In this section, we
introduce restrictions of $\decDNNF$ called certified $\decDNNF$ for which one
can check whether a CNF formula entails the certified $\decDNNF$. The idea is to
add information on $0$-sink to explain which clause would be violated by an
assignment leading to this sink.

Our inspiration comes from a known connection between regular resolution and
read once branching programs (\emph{i.e.} a $\decDNNF$ without
$\wedge$-gate~\cite{BeameLRS13}) that appears to be folklore but we refer the
reader to the book by Jukna~\cite[Section 18.2]{Jukna12} for a thorough and
complete presentation. It turns out that a regular resolution\footnote{A regular
  resolution proof is a resolution proof where, on each path, a variable is
  resolved at most once.} proof of unsatisfiability of a CNF $F$ can be
represented by a read once branching program $D$ whose sinks are labeled with
clauses of $F$. Moreover, for every $\tau$, if a sink labeled by a clause $C$ is
reached by a path compatible with $\tau$, then $C(\tau) = 0$. We generalize this
idea so that the function represented by a $\decDNNF$ is not only an
unsatisfiable CNF:

\begin{definition}
  \label{def:cafbdd}
  A \emph{certified $\decDNNF$ $D$ on variables $X$} is a $\decDNNF$ on variables
  $X$ such that every $0$-sink $\alpha$ of $D$ is labeled with a clause
  $C_\alpha$. $D$ is said to be \emph{correct} if for every $\tau \in \{0,1\}^X$
  such that there is a path from the source of $D$ to a $0$-sink $\alpha$
  compatible with $\tau$, $C_\alpha(\tau) = 0$.
\end{definition}

Given a certified $\decDNNF$, we denote by $Z(D)$ the set of $0$-sinks of $D$ and
by $F(D) = \bigwedge_{\alpha \in Z(D)} C_\alpha$.

Intuitively, the clause labeling a $0$-sink is an explanation on why one
assignment does not satisfy the circuit. The degenerated case where there are
only $0$-sinks and no $\wedge$-gates corresponds to the characterization of
regular resolution. 

A crucial property of certified $\decDNNF$ is that their correctness can be tested
in polynomial time:

\begin{theorem}
  \label{thm:check_correctness}
  Given a certified $\decDNNF$ $D$, one can check in polynomial time whether $D$
  is correct.
\end{theorem}
\begin{proof}
  By definition, $D$ is not correct if and only if there exists a $0$-sink
  $\alpha$, a literal $\ell$ in $C_\alpha$, an assignment $\tau$ such that
  $\tau(\ell) = 1$ and a path in $D$ from the source to $\alpha$ compatible with
  $\tau$. By Observation~\ref{obs:patharecompatible}, it is equivalent to the
  fact that there exists a path from the source to $\alpha$ that: either does
  not test the underlying variable of $\ell$ or contains the outgoing edge
  corresponding to $\tau(\ell) = 1$ when the underlying variable of $\ell$ is
  tested.

  In other words, $D$ is correct if and only if for every $0$-sink $\alpha$ and
  for every literal $\ell$ of $C_\alpha$ with variable $x$, every path from the
  source to $\alpha$ tests variable $x$ and contains the outgoing edge
  corresponding to an assignment $\tau$ such that $\tau(\ell) = 0$.

  This can be checked in polynomial time. Indeed, fix a $0$-sink $\alpha$ and a
  literal $\ell$ of $C_\alpha$. For simplicity, we assume that $\ell = x$ (the
  case $\ell = \neg x$ is completely symmetric). We have to check that every
  path from the source to $\alpha$ contains a decision node $\beta$ on variable
  $x$ and contains the outgoing edge of $\beta$ labeled with $0$. To check this,
  it is sufficient to remove all the edges labeled with $0$, going out of a
  decision node on variable $x$ and test that the source and $\alpha$ are now in
  two different connected components of $D$, which can obviously be done in
  polynomial time. Running this for every $0$-sink $\alpha$ and every literal
  $\ell$ of $C_\alpha$ gives the expected algorithm.
\end{proof}

The clauses labeling the $0$-sinks of a correct certified $\decDNNF$ naturally
connect to the function computed by $D$:

\begin{theorem}
  \label{thm:implication} Let $D$ be a correct certified $\decDNNF$ on variables
  $X$. We have $F(D) \Rightarrow D$.
\end{theorem}
\begin{proof}
  Observe that $F(D) \Rightarrow D$ if and only if for every $\tau \in
  \{0,1\}^X$, if $\tau$ does not satisfy $D$ then $\tau$ does not satisfy
  $F(D)$. Now let $\tau$ be an assignment that does not satisfy $D$. By
  definition, there exists a path compatible with $\tau$ from the source of $D$
  to a $0$-sink $\alpha$ of $D$. Since $D$ is correct, $C_\alpha(\tau) = 0$.
  Thus, $\tau$ does not satisfy $F(D)$ as $C_\alpha$ is by definition a clause
  of $F(D)$.
\end{proof}

\begin{corollary}
  \label{cor:implication} Let $F$ be CNF formula and $D$ be a correct certified
  $\decDNNF$ such that every clause of $F(D)$ are also in $F$. Then $F \Rightarrow
  D$.
\end{corollary}

\subsection{Proof systems}
\label{sec:pc}

\paragraph{Proof system for $\sSAT$.} One can use certified $\decDNNF$ to define
a proof system for $\sSAT$. The \emph{Knowledge Compilation based Proof System
  for $\sSAT$}, $\kcps(\sSAT)$ for short, is defined as follows: given a CNF
$F$, a certificate that $F$ has $k$ satisfying assignments is a correct
certified $\decDNNF$ $D$ such that:
\begin{itemize}
\item every clause of $F(D)$ are clauses of $F$,
\item $D$ computes $F$ and has $k$ satisfying assignments.
\end{itemize}

To check a certificate $D$, one has to check that $D$ is equivalent to $F$ and
has indeed $k$ satisfying assignments, which can be done in polynomial time as
follows:
\begin{itemize}
\item Check that $D$ is correct, which is tractable by
  Theorem~\ref{thm:check_correctness}.
\item Check that $D \Rightarrow F$, which is tractable by
  Corollary~\ref{cor:cnfentailement} and that every clause of $F(D)$ are clauses
  of $D$. By Corollary~\ref{cor:implication}, it means that $D \Leftrightarrow
  F$.
\item Computes the number $k$ of solutions of $D$, which is tractable by
  Theorem~\ref{thm:count_afbdd}.
\item Returns $(F,k)$.
\end{itemize}

This proof system for $\sSAT$ is particularly well-suited for the existing tools
solving $\sSAT$ in practice. Many of them such as
\texttt{sharpSAT}~\cite{thurley06} or \texttt{cachet}~\cite{sang04} are based on
a generalization of DPLL for counting which is sometimes refered as exhaustive
DPLL in the literature. It has been observed by Huang and
Darwiche~\cite{HuangD05} that these tools were implicitly constructing a
$\decDNNF$ equivalent to the input formula. Tools such as
\texttt{c2d}~\cite{OztokD15}, \texttt{D4}~\cite{LagniezM17} or
\texttt{DMC}~\cite{lagniez2018dmc} already exploit this connection and have the
option to directly output an equivalent $\decDNNF$. These solvers explore the
set of satisfying assignments by branching on variables of the formula which
correspond to a decision node and, when two variable independent components of
the formula are detected, compute the number of satisfying assignments of both
components and take the product, which corresponds to a decomposable
$\wedge$-gate. When a satisfying assignment is reached, it corresponds to a
$1$-sink. If a clause is violated by the current assignment, then it corresponds
to a $0$-sink. At this point, the solvers could also label the $0$-sink by the
violated clause which would give a correct certified $\decDNNF$.

\paragraph{Proof system for $\mSAT$.} As for $\sSAT$, one can exploit the
tractability of many problems on $\decDNNF$ to define a proof system for
$\mSAT$. Given a CNF formula $F$, ket $\tilde{F} = \bigwedge_{C \in F} C \vee
\neg s_C$ be the formula where each clause is augmented with a fresh
\emph{selector} variable. Let $S = \{s_C \mid C \in F\}$. Observe that $M(F)$ is
exactly $\max_{\tau \models \tilde{F}} |\{s \in S \mid \tau(s)=1\}|$ since if
$\tau \models \tilde{F}$ and $\tau(s_C) = 1$, then $\tau \models C$. By
Theorem~\ref{thm:count_afbdd}, if $\tilde{F}$ is represented by a $\decDNNF$
$D$, then one can solve this problem in polynomial time in $\size(D)$. The proof
system $\kcps(\mSAT)$ is defined as follows: given a CNF $F$, a certificate is a
correct certified $\decDNNF$ $D$ with clauses in $\tilde{F}$ that computes
$\tilde{F}$. The proof may be checked as before by checking both the correctness
of $D$ and the fact that $D \Leftrightarrow \tilde{F}$. However, we are not
aware of any tool solving $\mSAT$ based on this technique and thus the
implementation of such a proof system in existing tools may not be realistic. It
will still be worth comparing this proof system with the resolution for
$\mSAT$~\cite{BonetLM07}.

In general, we observe that we can use this idea to build a proof system
$\kcps(Q)$ for any tractable problem $Q$ on $\decDNNF$. This could for example
be applied to weighted versions of $\sSAT$ and $\mSAT$.

\paragraph{Combining proof systems.} An interesting feature of $\kcps$-like
proof systems is that they can be combined with other proof systems for $\UNSAT$
to be made more powerful. Indeed, one could label the $0$-sink of the $\decDNNF$
with a clause $C$ that are not originally in the initial CNF $F$ but that is
entailed by $F$, that is, $F \Rightarrow C$. In this case,
Corollary~\ref{cor:implication} would still hold. The only thing that is needed
to obtain a real proof system is that a proof that $F \Rightarrow C$ has to be
given along the correct certified $\decDNNF$, that is, a proof of
unsatisfiability of $F \wedge \neg C$. Any proof system for $\UNSAT$ may be used
here.

\paragraph{Lower bounds.} Lower bounds on the size of $\decDNNF$ representing
CNF formulas may be directly lifted to lower bounds for $\kcps(\sSAT)$ or
$\kcps(\mSAT)$. There exists families of monotone $2$-CNF that cannot be
represented as polynomial size $\decDNNF$~\cite{BeameLRS13,BovaCMS16,CapelliPhd2016}. It
directly gives the following corollary:
\begin{corollary}
  There exists a family $(F_n)_{n \in \mathbb{N}}$ of monotone $2$-CNF such that
  $F_n$ is of size $O(n)$ and any proof for $F_n$ in $\kcps(\sSAT)$ and
  $\kcps(\mSAT)$ is of size at least $2^{\Omega(n)}$.
\end{corollary}
An interesting open question is to find CNF formulas having polynomial size
$\decDNNF$ but no small proof in $\kcps(\sSAT)$.

\section{Future work}
\label{sec:conclusion}

In this paper, we have developed techniques based on circuits used in knowledge
compilation to extend existing proof systems for tautology to harder problems.
It seems possible to implement these systems into existing tools for $\sSAT$
based on exhaustive DPLL, which would allow these tools to provide an
independently checkable certificate that their output is correct, the same way
$\SAT$-solvers returns a proof on unsatisfiable instances. It would be
interesting to see how adding the computation of this certificate to existing
solver impacts their performances. Another interesting direction would be to
compare the power of $\kcps(\mSAT)$ with the resolution for $\mSAT$ of Bonet et
al.~\cite{BonetLM07} and to see how such proof systems could be implemented in
existing tools for $\mSAT$. Finally, we think that a systematic study of other
languages used in knowledge compilation such as deterministic DNNF should be
done to see if they can be used as proof systems, by trying to add explanations
on why an assignment does not satisfy the circuit.

\newpage

\bibliography{biblio}

\end{document}